\documentclass[aps,reprint,floatfix]{revtex4-1}
\usepackage{amsmath,amssymb,graphicx}
\usepackage{hhline,slashbox,bbm}

\makeatletter
\DeclareRobustCommand{\cev}[1]{%
  \mathpalette\do@cev{#1}%
}
\newcommand{\do@cev}[2]{%
  \fix@cev{#1}{+}%
  \reflectbox{$\m@th#1\vec{\reflectbox{$\fix@cev{#1}{-}\m@th#1#2\fix@cev{#1}{+}$}}$}%
  \fix@cev{#1}{-}%
}
\newcommand{\fix@cev}[2]{%
  \ifx#1\displaystyle
  \mkern#23mu
  \else
  \ifx#1\textstyle
  \mkern#23mu
  \else
  \ifx#1\scriptstyle
  \mkern#22mu
  \else
  \mkern#22mu
  \fi
  \fi
  \fi
}

\makeatother

\newcounter{theoremcounter}
\stepcounter{theoremcounter}

\newtheorem{theorem}{Theorem}

\newenvironment{proof}[1][Proof]{\begin{trivlist}
\item[\hskip \labelsep {\bfseries #1}]}{\end{trivlist}}

\newcommand{\qed}{\hfill $\blacksquare$}

\newcommand{\bs}{\boldsymbol}
\newcommand{\bra}[1]{\left\langle #1\right|}
\newcommand{\ket}[1]{\left|#1\right\rangle}

\newcommand{\expval}[3]{\left\langle #1\middle|#2\middle|#3\right\rangle}
\newcommand{\ketbra}[2]{\ket{#1}\bra{#2}}

\makeatletter
\newcommand{\vast}{\bBigg@{4}}
\newcommand{\Vast}{\bBigg@{5}}
\makeatother

\begin{document}

\title{Measurement Contextuality and Planck's Constant}
\author{Lucas Kocia}
\affiliation{Department of Physics, Tufts University, Medford, Massachusetts 02155, U.S.A.}
\affiliation{National Institute of Standards and Technology, Gaithersburg, MD, 20899}
\author{Peter Love}
\affiliation{Department of Physics, Tufts University, Medford, Massachusetts 02155, U.S.A.}
\begin{abstract}
  Contextuality is a necessary resource for universal quantum computation and non-contextual quantum mechanics can be simulated efficiently by classical computers in many cases. Orders of Planck's constant, \(\hbar\), can also be used to characterize the classical-quantum divide by expanding quantities of interest in powers of \(\hbar\)---all orders higher than \(\hbar^0\) can be interpreted as quantum corrections to the order \(\hbar^0\) term. We show that contextual measurements in finite-dimensional systems have formulations within the Wigner-Weyl-Moyal (WWM) formalism that require higher than order \(\hbar^0\) terms to be included in order to violate the classical bounds on their expectation values. As a result, we show that contextuality as a resource is equivalent to orders of \(\hbar\) as a resource within the WWM formalism. This explains why qubits can only exhibit state-independent contextuality under Pauli observables as in the Peres-Mermin square while odd-dimensional qudits can also exhibit state-dependent contextuality. In particular, we find that qubit Pauli observables lack an order \(\hbar^0\) contribution in their Weyl symbol and so exhibit contextuality regardless of the state being measured. On the other hand, odd-dimensional qudit observables generally possess non-zero order \(\hbar^0\) terms, and higher, in their WWM formulation, and so exhibit contextuality depending on the state measured: odd-dimensional qudit states that exhibit measurement contextuality have an order \(\hbar^1\) contribution that allows for the violation of classical bounds while states that do not exhibit measurement contextuality have insufficiently large order \(\hbar^1\) contributions.
\end{abstract}
\maketitle

\section{Introduction}
\label{sec:intro}

The Wigner-Weyl-Moyal (WWM) formalism is a particularly powerful representation of quantum mechanics based on quasi-probability functions. Starting from Wootters' original derivation of discrete Wigner functions~\cite{Wootters87}, there has been much work on analyzing states and operators in finite Hilbert spaces by considering their quasiprobability representation on continuous and discrete support~\cite{Berezin77,Varilly89,Leonhardt95,Rivas99,Heiss00,Bianucci02,Ruzzi05,Veitch12,Mari12,Marchiolli12,Gross06,Gross08,Ferrie08,Ferrie09,Ferrie10,Ferrie11,Wallman12,Howard14}. The WWM formalism gives a discrete quasiprobability representation in terms of a classical phase space and explicitly introduces classical mechanics and quantum corrections in terms of higher order corrections with respect to \(\hbar\).

Contextuality is a quantum resource whose presence in an experiment has been shown to be equivalent to negativity in the associated discrete Wigner functions and Weyl symbols of the states, operations, and measurements that are involved~\cite{Veitch12, Mari12}. Preparation contextuality~\cite{Spekkens08} has been shown to be equivalent to the necessity of including higher than order \(\hbar^0\) terms in path integral expansions of unitary operators in the WWM formalism for proper computation of propagation~\cite{Kocia16,Kocia17_2}. As a result, stabilizer state propagation under Clifford gates can be shown to be non-contextual since it does not require higher than \(\hbar^0\) terms within WWM, or equivalently, the Wigner function of stabilizer states are non-negative and the Weyl symbols of Clifford gates are positive maps~\cite{Rivas99,Rivas00,Gross08}.

Here we complete the characterization of the relationship between orders of \(\hbar\) in WWM and contextuality by showing that measurement contextuality~\cite{Spekkens05,Spekkens08} is also equivalent to non-classicality as dictated by powers of \(\hbar\). Specifically, contextuality requires us to include higher than order \(\hbar^0\) terms in the \(\hbar\) expansion of observables within the WWM formalism to obtain expectation values that violate classical bounds. Moreover, this equivalency explains why qubits exhibit state-independent contextuality while odd-dimensional qudits also exhibit state-dependent contextuality. Furthermore, this confirms the result that any quasi-probability representation that defines a set of complete experimental configurations, such as the WWM, must either exhibit negativity in the quasiprobability functions for either states or measurements, or make use of a deformed probability calculus~\cite{Ferrie08} (defined below).

Let us begin by defining the context of a measurement. A projection of a quantum state onto a rank \(n \ge 2\) subspace of its Hilbert space can be decomposed into a sum of smaller rank projectors in many ways. Fixing a subset of the terms in a sum of such projectors, there are many choices for the remaining terms. Each non-commuting decomposition of the remaining terms corresponds to a ``context'' of the measurement.

Instead of projectors we may speak instead about observables. The rank \(n \geq 2\) subspace is then a degenerate eigenspace of some observable. The different contexts correspond to different choices of complete sets of commuting observables whose eigenstates are the projectors onto the different contexts. Again, different choices of complete sets of commuting observables do not commute with each other, for otherwise they would share the same eigenstates and so correspond to the same context.

For instance, considering two qubits, the measurement of \(\hat X \hat I\) corresponds to a projection onto a subspace of rank two. It can be performed in the context of \(\{\hat X \hat I,\, \hat I \hat X\}\) or \(\{\hat X \hat I,\, \hat I \hat Z\}\). The operators in each set commute with each other. However, the two operators that distinguish these contexts, \(\hat I \hat X\) and \(\hat I \hat Z\) anticommute, and the product of the operators in each set anticommute with each other~\cite{Peres91}. Hence the outcome of a measurement of \(\hat X \hat I\) is not independent of the choice of context. Each set corresponds to a projection onto the full rank four Hilbert space and is a separate context for \(\hat X \hat I\).

These two sets correspond to the first row and column of the Peres-Mermin square shown in Table~\ref{tab:peresmerminsq} (the third element in the row and/or column is redundant---its outcome is determined by the first two measurements)~\cite{Peres91}.
\begin{table}[ht]
  \begin{tabular}{|c|c|c|}
    \hline
    \(\hat X \hat I\) & \(\hat I \hat X\) & \(\hat X \hat X\)\\
    \hline
    \(\hat I \hat Z\) & \(\hat Z \hat I\) & \(\hat Z \hat Z\)\\
    \hline
  \(\hat X \hat Z\) & \(\hat Z \hat X\) & \(\hat Y \hat Y\)\\
  \hline
  \end{tabular}
  \label{tab:peresmerminsq}
\caption{The Peres-Mermin square.}
\end{table}
Every element in the table contains a Pauli observable. The operators in the same row or column commute and so can be performed independently of each other. However, operators in a particular row or column do not commute with operators in other rows and columns. While every observable in the table has \(\pm 1\) as possible outcomes, the row and columns are completely-commuting sets of operators (CCSOs) whose product have \(+1\) as their only possible outcomes due to the Pauli operator relation \(\hat \sigma_j \hat \sigma_k \hat \sigma_l = i\epsilon_{jkl} +\delta_{jk} \hat \sigma_l\). The exception is the third column, which has \(-1\) as its product's only possible outcome as can be seen from the same identity relation.

As a result, the different contexts of a particular observable in the Peres-Mermin square correspond to its row and column. Assigning classical \(\pm 1\) outcomes to the measurement of the operators in the table such that their product satisfies the constraints on the outcomes of the rows and columns is impossible. Given one of the measurements in the Peres-Mermin square, it is impossible to assign an outcome without specifying the context of the full rank observables, i.e. whether the measurement is taken column-wise or row-wise.

In this way, the Peres-Mermin square demonstrates that qubits exhibit measurement contextuality with Pauli operators~\cite{Kochen67,Redhead87,Peres91}. Furthermore, the Peres-Mermin square shows that qubits exhibit state-independent contextuality; no matter what qubit state the Peres-Mermin measurements are performed on, their outcome will always depend on the context or which other measurements are co-performed~\cite{Mermin93}.

On the other hand, higher-dimensional qudits can also exhibit measurement contextuality that is state-dependent. For odd \(d\)-dimensional qudits, it is possible to exhibit contextuality for a single qudit. In particular, Klyachko \emph{et al}.~\cite{Klyachko08} developed a scheme (called KCSB after the four authors) for a single qutrit (\(d=3\)) that exhibits state-dependent contextuality~\cite{Howard13}. KCSB and other constructions are all different demonstrations of the Kochen-Specker theorem~\cite{Kochen67}, which was also originally proposed for a qutrit, and exclude any hidden variable theory (HVT) that is non-contextual from reproducing quantum mechanics for \(d \ge 3\).

The KCSB scheme is defined as follows. Consider a set \(\Gamma\) of five rank-\(1\) projectors, \(\Gamma = \{\hat \Pi_i\}_{i=1,\ldots,5}\), which commute with neighboring pairs (\([\hat \Pi_i, \hat \Pi_{i\oplus 1}]=0\), where \(i \oplus 1\) denotes addition performed modulo \(5\)). We further impose that commuting projectors are also orthogonal to each other, i.e. they both cannot take on the value \(+1\). Such a realization is illustrated in Fig.~\ref{fig:KCSB}, where the five projectors are placed along the points of a pentagon such that projectors that share an edge commute. The two contexts for a measurement on any vertex of the pentagon corresponds to its two adjoining edges; i.e. \(\hat \Pi_1\) can be measured in the context of \(\{\hat \Pi_1, \hat \Pi_2\}\) or \(\{\hat \Pi_1, \hat \Pi_5\}\).

Finally, we define,
\begin{equation}
  \hat \Sigma_{\Gamma} = \sum_{\hat \Pi_i \in \Gamma} \hat \Pi_i.
\end{equation}
These five projectors obey \(\hat \Pi_i \hat \Pi_{i \oplus 1} = 0\). By Mermin's argument, any classical outcome \(\{0,+1\}\) one preassigns to the measurement must obey the same relationship~\cite{Mermin93}. Adjacent vertices therefore cannot both be assigned the outcome \(+1\), and so the maximum expectation value for \(\hat \Sigma_\Gamma\) is two:
\begin{equation}
  \expval{\Psi}{\hat \Sigma_{\Gamma}}{\Psi}_\text{CM} \le 2.
\end{equation}

This upper bound is higher if the above expectation value is evaluated quantum mechanically. Namely, of the three eigenstates of \(\hat \Sigma_{\Gamma}\), \(\ket{\phi_1}\), \(\ket{\phi_2}\), and \(\ket{\phi_3}\), with eigenvalues \((5 - \sqrt 5)/2 \approx 1.382\), \(1.382\) and \(\sqrt{5} \approx 2.236\) respectively, the last state can be shown to saturate the quantum bound~\cite{Klyachko08}:
\begin{equation}
  \expval{\Psi}{\hat \Sigma_{\Gamma}}{\Psi}_\text{QM} \le \sqrt{5} \approx 2.236.
\end{equation}

As a result, \(\hat \Sigma_\Gamma\) can be interpreted as a witness for contextuality; any state \(\Psi\) with expectation value \(\expval{\Psi}{\hat \Sigma_\Gamma}{\Psi} > 2\) exhibits measurement contextuality. The states \(\ket{\phi_1}\) and \(\ket{\phi_2}\) don't exhibit contextuality within the KCSB construction while \(\ket{\phi_3}\) does.

An equivalent witness for contextuality that will be more useful for us later is the sum of products of non-commuting pairs of observables. Defining this set as \(\Gamma^2\), where
\begin{eqnarray}
  \Gamma^2 &=& \left\{\hat \Pi_1 \hat \Pi_3, \hat \Pi_1 \hat \Pi_4, \hat \Pi_2 \hat \Pi_4, \hat \Pi_2 \hat \Pi_5, \hat \Pi_3 \hat \Pi_5,\right.\\
           && \left.\hat \Pi_3 \hat \Pi_1, \hat \Pi_4 \hat \Pi_1, \hat \Pi_4 \hat \Pi_2, \hat \Pi_5 \hat \Pi_2, \hat \Pi_5 \hat \Pi_3\right\},
\end{eqnarray}
this witness for contextuality can be written as
\begin{equation}
  \hat \Sigma_{\Gamma^2} = \sum_{\hat \Pi_i \hat \Pi_j \in \Gamma^2} \hat \Pi_i \hat\Pi_j.
\end{equation}

Following the conclusion from the \(\Sigma_\Gamma\) witness for contextuality---that a maximum of two non-adjacent observables can be assigned outcomes of \(+1\)---it follows that the maximum classical expectation value for \(\hat \Sigma_{\Gamma^2}\) is also two:
\begin{equation}
  \expval{\Psi}{\hat \Sigma_{\Gamma^2}}{\Psi}_\text{CM} \le 2.
\end{equation}

Again, the upper bound is higher in quantum mechanics. While \(\ket{\phi_1}\) and \(\ket{\phi_2}\) have expectation values with \(\hat \Sigma_{\Gamma^2}\) of \(0.263932<1\), \(\ket{\phi_3}\) saturates the quantum bound:
\begin{equation}
  \expval{\Psi}{\hat \Sigma_{\Gamma^2}}{\Psi}_\text{QM} \le 5 - \sqrt{5} \approx 2.76393.
\end{equation}
As a result, \(\hat \Sigma_{\Gamma^2}\) also exhibits measurement contextuality with the state \(\ket{\phi_3}\).

Just as for the Peres-Mermin square, any effort to assign more than two classical \(+1\) outcomes to the measurements in the KCSB pentagon, while satisfying the constraint that the neighbors have different outcomes, is impossible. The outcome of any measurement on a vertex depends on whether it is taken in the context of its left or right neighboring observable. However, unlike for the Peres-Mermin square, this dependence on context only holds for certain quantum states.

\begin{figure}[ht]
  \includegraphics[scale=1.0]{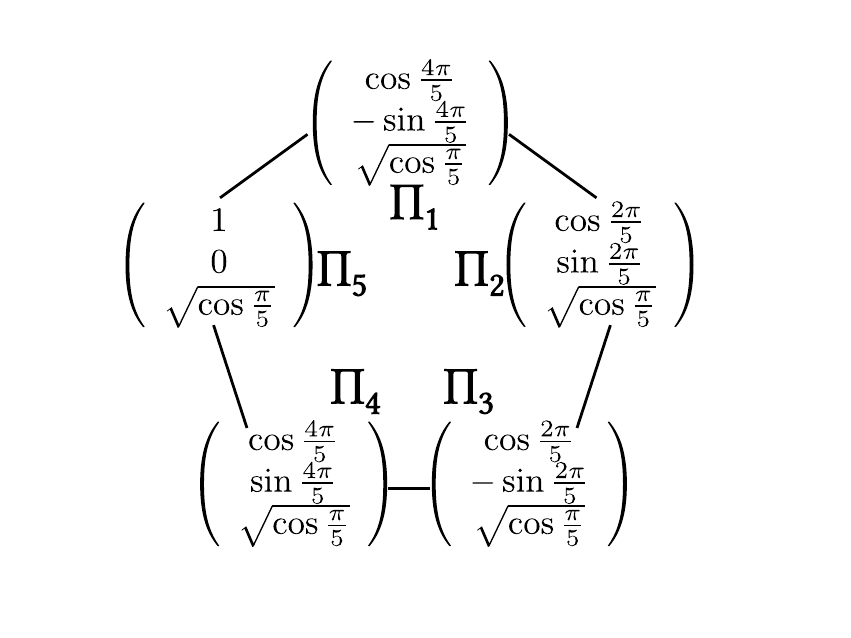}
  \caption{The KCSB contextuality construction for a qutrit. The five \(\Pi_i\) projectors are outer products of the vectors above (after normalization) and commute with each other if they share an edge.}
  \label{fig:KCSB}
\end{figure}

The state dependence of measurement contextuality for qutrits is not limited to the KCSB scheme. In general, odd-dimensional qudits can exhibit state-dependent contextuality under Pauli measurements, while qubits can only exhibit state-independent contextuality. This is a curious dichotomy.  It has been found that (preparation) contextuality is a necessary resource for universal quantum computation~\cite{Howard14} and so is of great interest for the development of quantum computers. Does this suggest that qubits are somehow harder to simulate classically compared to odd \(d>2\) qudits? This is certainly not the case for the simulation of qubit stabilizer states. Preparation and evolution of these states is efficient using the Aaronson-Gottesmann algorithm and these operations are also noncontextual for qubits~\cite{Kocia17_2}. This is the puzzle we examine in the present paper.

Another way to characterize a divide between the quantum characteristics of an observable in its different contexts and its classical limit is to consider the observable expressed as products of other observables that define a particular context. It is then instructive to examine these expressions' expansions in terms of Planck's constant \(\hbar\). Since the classical limit is reached as \(\hbar \rightarrow 0\), the leading order \(\hbar^0\) term, which is independent of the magnitude of \(\hbar\), can be interpreted as the classical term while higher order terms act as quantum corrections~\cite{Tannor07}.

Applying such an approach requires a path integral formulation of discrete quantum systems. We will use the Wigner-Weyl-Moyal (WWM) formalism which has been developed for odd \(d\)-dimensional qudits~\cite{Wootters87,Rivas99,Rivas00,Gross08} and was recently extended to qubits (\(d=2\))~\cite{Kocia17_2}. The WWM formalism is particularly useful for finite-dimensional systems because it uses the conjugate degrees of freedom of ``chords'' and ``centers'' to define Hamiltonian phase space, instead of momentum and position as is the traditional approach for infinite-dimensional continuous Hilbert spaces. The former are associated with translation and reflection operators that retain their role in a well-defined Lie group in a finite-dimensional Hilbert space, while the latter are associated with momentum and position operators which no longer form a simple Lie algebra in finite-dimensional spaces.

The WWM formalism allows for a semiclassical expansion of unitary operations in discrete systems in powers of \(\hbar\) through the van Vleck propagator approach~\cite{Almeida98}. It is also possible to semiclassically expand Hermitian operators (observables) in powers of \(\hbar\) using the Groenewold Rule~\cite{Groenewold46}, and we will show that this also holds true for the discrete WWM formalism for both \(d=2\) and odd \(d\). By considering observables that are also products of observables (corresponding to different contexts), we show that the WWM formalism can be used to formulate an equivalent statement of measurement contextuality with regards to orders of \(\hbar\).

 This can be restated in the somewhat more general language of quasi-probability distributions and frame representations of quantum mechanics. Given a measurable space \(\Gamma\) (a phase space) endowed with a positive measure \(\mu\) and a measurement that is represented by a set of conditional properties \(\left\{M_k(\alpha) \in \mathbb R\right\}\) that all satisfy some natural properties~\cite{Ferrie08}, the probability of obtaining outcome \(k\) is \(P[\psi](k) = \int_\Gamma \text d \mu(\alpha,\beta) \rho_\psi(\alpha) M_k(\beta) \left< \hat E(\alpha), \hat E(\beta) \right>\), where \(\hat E\) is any frame dual to \(\hat F\) (defined in~\cite{Ferrie08}). This quantity is a deformation of the classical probability function: \(P(k) = \int_\Gamma \text d \mu(\alpha) \rho(\alpha) M_k(\alpha)\).

We shall see that the WWM formalism reduces to such a simpler classical form for non-contextual measurements. In this case, the WWM formalism is readily seen to be a non-contextual HVT where every state has an associated ensemble of ontic states that are deterministically measured independently of their context. Observables on states that produce non-contextual outcomes can be found to reduce to indicator functions of an associated ensemble of ontic states that do not rely on a deformed probability calculus in the sense of~\cite{Ferrie08} to violate their classical bounds, while states that produce contextual outcomes do rely on a deformed probability calculus introduced by the higher than order \(\hbar^0\) corrections.

\section{A Summary of the Wigner-Weyl-Moyal Formalism}

The WWM formalism was originally developed to try to introduce classical phase space into quantum statistical mechanics~\cite{Weyl27,Wigner32,Moyal49} but has also found much use in analyzing general quantum mechanical phenomena. The formalism's relationship with classical mechanics was formalised with the introduction of a center-chord reformulation that allowed for classical trajectories to be variationally expanded with respect to \(\hbar\)~\cite{Berry77,Marinov80,Marinov79,Berezin77,Almeida92,Almeida90}. 

Taken as whole, the WWM formalism is a faithful representation of quantum mechanics, like the position (or computational basis) matrix representation of quantum operators and states. As such, the it reproduces all of the results of quantum mechanics. Operators \(\hat O\) are replaced by their Weyl symbols, \(W_O(\bs \zeta)\), which are functions on phase space. The Weyl symbols of states \(\ket \Psi\) are often called Wigner functions and the expectation value of an operator \(\hat O\) in a state \(\ket \Psi\) is given by:
\begin{equation}
  \expval{\Psi}{\hat O}{\Psi} = \int W_O(\bs \zeta) W_\Psi(\bs \zeta) \text d^n \bs \zeta.
\end{equation}

The WWM formalism was originally developed for continuous infinite-dimensional systems, and has since been extended to finite-dimensional qudits~\cite{Rivas99,Rivas00,Kocia17_2}. This extension differs depending on whether the qudits are odd-dimensional or even-dimensional (\(d=2\)). In the former case, the continuous WWM formalism can be periodized and discretized by setting \(\hbar = d/2 \pi\), so that the phase space is generated by two coordinates---momenta and positions, \(\bs \zeta \equiv \bs x = (p,q)\), which take finite scalar values. In the latter, a Grassmann algebra needs to be used such the phase space is generated by three coordinates, \(\bs \zeta \equiv \bs \xi = (\xi_p, \xi_q, \xi_r)\), which are Grassmann elements and not scalars.

Here we present a summary of the WWM formalism for odd \(d\) qudits and \(d=2\) qubits.

\subsection{Odd-Dimensional Qudits}
\label{sec:odddqudits}
To define Weyl phase space, we begin in infinite-dimensional space and define the Weyl-Heisenberg operators~\cite{Almeida98}:
\begin{equation}
  \label{eq:weylheisenberg}
  \hat T(\bs \lambda_p, \bs \lambda_q) = \exp \left(-\frac{i}{2\hbar} \bs \lambda_p \cdot \bs \lambda_q\right) \hat Z^{\bs \lambda_p} \hat X^{\bs \lambda_q}.
\end{equation}
The elements of the set \(\hat T\) are Hilbert-Schmidt orthogonal.
\(\hat Z\) and \(\hat X\) generate a Lie group and correspond to the ``boost'' operator:
\begin{equation}
  \hat Z^{\delta p} \ket{q'} = e^{\frac{i}{\hbar} \hat q \delta p} \ket{q'} = e^{\frac{i}{\hbar} q' \delta p} \ket{q'},
\end{equation}
and the ``shift'' operator:
\begin{equation}
  \hat X^{\delta q} \ket{q'} = e^{-\frac{i}{\hbar} \hat p \delta q} \ket{q'} = \ket{q' + \delta q},
\end{equation}
which satisfy the Weyl relation:
\begin{equation}
  \label{eq:contWeylrelation}
  \hat Z \hat X =  e^{\frac{i}{\hbar}} \hat X \hat Z.
\end{equation}

We define \(\hat R(\bs x)\) as the symplectic Fourier transform of \(\hat T(\bs \lambda)\):
\begin{eqnarray}
  \label{eq:contreflection}
  &&\hat R(\bs x_p, \bs x_q) = \left(2 \pi \hbar\right)^{-n} \int^\infty_{-\infty} \mbox{d} \bs \lambda e^{\frac{i}{\hbar} {\bs \lambda}^T \bs{\mathcal J} {\bs x} } \hat T(\bs \lambda),
\end{eqnarray}
where
\begin{equation}
\bs{\mathcal J} =  \left( \begin{array}{cc} 0 & -\mathbb{I}_{n}\\ \mathbb{I}_{n} & 0 \end{array}\right),
\end{equation}
for \(\mathbb{I}_n\) the \(n\)-dimensional identity. The domain of the \(\hat R\) operators, \(\bs x\), can be associated with phase space.

The Weyl symbol of operator \(\hat \rho\) can be expressed as the coefficient of the operator expanded in the basis of states \(\hat R(\bs x_p, \bs x_q)\), which are parametrized by \(\bs x\):
\begin{equation}
  \hat \rho = \int^{\infty}_{-\infty} \text d \bs x \hat {R}^\dagger(\bs x_p, \bs x_q) W_{\hat \rho}(\bs x_p, \bs x_q).
\end{equation}
If \(\hat \rho\) is a state, \(W_{\hat \rho}(\bs x)\) is the corresponding Wigner function.

Restricting this to finite odd \(d\)-dimensional systems involves setting \(\hbar = d/2\pi\), and enforcing periodic boundary conditions~\cite{Rivas99}. The points \((\bs \lambda_p, \bs \lambda_q)\) and \((\bs x_p, \bs x_q)\) become elements in \((\mathbb{Z} / d \mathbb{Z})^{2n}\) and form a discrete ``web'' or Weyl ``grid''. The generalized translation operator becomes
\begin{equation}
\hat {T}(\bs \lambda_p, \bs \lambda_q) = \omega^{-\bs \lambda_p \cdot \bs \lambda_q (d+1)/2} \hat {Z}^{ \bs \lambda_p} \hat {X}^{ \bs \lambda_q}.
\end{equation}
where \(\omega\equiv \exp 2 \pi i/d\) and \((d+1)/2\) is equivalent to \(1/2\) in mod odd-\(d\) arithmetic, and the reflection operator becomes:
\begin{equation}
\hat {R}(\bs x_p, \bs x_q) = d^{-n} \sum_{\substack{\bs \xi_p, \bs \xi_q \in \\ (\mathbb{Z} / d \mathbb{Z})^{ n}}} e^{\frac{2 \pi i}{d} (\bs \xi_p, \bs \xi_q) \bs{\mathcal J} (\bs x_p, \bs x_q)^T} \hat {T}(\bs \xi_p, \bs \xi_q).
\end{equation}

Again, the Weyl symbol of an operator \(\hat \rho\) can be expressed as a (now discrete) coefficient of the density matrix expanded in the basis of states \(\hat R(\bs x_p, \bs x_q)\):
\begin{equation}
  \hat \rho = \sum_{\substack{\bs \xi_p, \bs \xi_q \in \\ (\mathbb{Z} / d \mathbb{Z})^{ n}}} \hat {R}^\dagger(\bs x_p, \bs x_q) W_{\hat \rho}(\bs x_p, \bs x_q).
\end{equation}

If we regard the points of discrete phase space as ontic states labelled by \((\bs p, \bs q)\), as we do in classical mechanics, then any classical operation maps an ontic state \((\bs p, \bs q)\) to a new ontic state \((\bs p', \bs q')\). Evidently, these ontic states are not allowed states of quantum mechanics. However, the set of quantum states called stabilizer states are represented by non-negative Wigner functions that faithfully represent a subset of quantum states~\cite{Gross08} and we can interpret these non-negative states as representing an ensemble of ontic states. That is, in each realization of an experiment there is some true ontic state \((\bs p, \bs q)\) present initially, but repeated measurements can only sample these states from the distribution implied by \(W_\psi\). Furthermore, Clifford gates are a subset of all quantum operations that take stabilizer states to other stabilizer states. Thus, Clifford gates can be interpreted to take ontic states to other ontic states.

This interpretation implies that in the preparation and evolution parts of an experiment utilizing only stabilizer states and Clifford gates there is a real ontic state present in each realization of the experiment. The preparation samples an ontic state from the distribution \(W_\psi\) and the Clifford gates then deterministically map this ontic state to another ontic state.

Just as in Bohmian mechanics or in Bell's HVTs for a single spin one-half~\cite{Bell66}, the WWM formalism augments the wave function (which determines the Wigner function) by an ontic state---in this case of a particle in a discrete phase space. One may therefore regard the theory as ontological---there is always a real state actually present, with an epistemic restriction imposed by the set of allowed Wigner functions. The probabilistic nature of the theory therefore arises from this epistemic restriction, i.e. from our enforced ignorance of what the true ontic state actually is. This is of course the same situation as for classical probabilistic theories.

How do we regard measurement in this interpretation of the WWM formalism? Absent the epistemic restriction imposed by the wavefunction, one would simply like to measure the ontic state \((\bs p, \bs q)\) as one would for a classical particle. Of course, because the WWM formalism is a faithful representation of quantum mechanics this is impossible. Let's consider measurements of \(\bs p\) and \(\bs q\) separately. Given \(W_\psi(\bs x_p, \bs x_q)\) there is an implied marginal distribution of \(\bs x_p\) or \(\bs x_q\):
\begin{equation}
P[\psi](\bs x_p)=\sum_{\bs x_q'}W_\psi(\bs x_p, \bs x_q')=\sum_{\bs x_p', \bs x_q'}\delta_{\bs x_p, \bs x_p'}W_\psi(\bs x_p', \bs x_q').
\end{equation}
and:
\begin{equation}
Q[\psi](\bs x_q)=\sum_{\bs x_p'}W_\psi(\bs x_p', \bs x_q)=\sum_{\bs x_p', \bs x_q'}\delta_{\bs x_q, \bs x_q'}W_\psi(\bs x_p', \bs x_q')
\end{equation}

As usual in HVTs, we can therefore interpret measurements in terms of indicator functions on the discrete phase space. In this simple example, \(\delta_{\bs x_p, \bs x_p'}\) and \(\delta_{\bs x_q, \bs x_q'}\). How do these indicator functions arise in the WWM formalism given either a Hermitian observable or more generally a POVM for the measurement?

The POVM representation of the measurement of a (Hermitian) observable is just given by the set of orthogonal projectors onto a complete basis of distinct eigenvectors. The Weyl symbols of these projectors are of course just the Wigner functions of the corresponding state, which are normalized characteristic functions of the support of those states. In summary, the indicator functions for a general POVM for a measurement corresponds to the Weyl symbols of the associated orthogonal set of projectors.

On the other hand, the indicator functions of an observable \(\hat A\) in WWM is simply its Weyl symbol \(W_A(\bs x)\). In this paper, we are interested in observables that are expressed as products with other observables (e.g. \(\hat A \hat B\)) to indicate the contexts of the measurement.

In the continuous case, it can be shown by Groenewold's Rule~\cite{Groenewold46} that the Weyl symbol of a product of operators is equal to the product of their Weyl symbols up to order \(\hbar^0\)~\cite{Almeida98}:
\begin{equation}
  W_{A B}(\bs x) = W_A(\bs x) W_B(\bs x) + \mathcal O(\hbar).
\end{equation}
This identity remains true after discretization and periodization for an odd \(d\) dimensional system.

This means that in the classical limit (when \(\hbar \rightarrow 0\) and the order \(\hbar^0\) terms are the only ones remaining non-zero), the Weyl symbol of a product of observables is the product of the Weyl symbols of the observables. In other words, the indicator functions of the observables become independent of each other---an observation that will prove very useful when we study measurement contextuality in Section~\ref{sec:meascontext}.

\subsection{Qubits}
\label{sec:qubits}
While the WWM formalism for odd-dimensional qudits can be made with the two generators, \(\hat p\) and \(\hat q\), the WWM formalism for qubits requires three generators. This is because the translation operator forms a subgroup of \(SU(d)\) only for odd \(d\)~\cite{Bengtsson17}, and Clifford gates in any odd prime power dimension are unitary two-designs, while multi\emph{qubit} Clifford gates are also unitary three-designs~\cite{Zhu15,Wallman12}.

Let \(\xi_p\), \(\xi_q\) and \(\xi_r\) be three real generators of a Grassmann algebra \(\mathcal G_3\). Hence,
\begin{equation}
\xi_j \xi_k + \xi_k \xi_j \equiv \{\xi_j, \xi_k\} = 0,  \quad \text{for} \, j,k \in \{1,2,3\},
\end{equation}
where we can identify \(\xi_p \equiv \xi_1\), \(\xi_q \equiv \xi_2\) and \(\xi_r \equiv \xi_3\).

The three real generators can be treated as classical canonical variables:
\begin{equation}
  i \sum_{j} \left(\xi_k \frac{\cev \partial}{\partial \xi_j}\right) \left(\frac{\vec \partial}{\partial \xi_j} \xi_l\right) = \{\xi_k, \xi_l\}_{\text{P.B.}} = i \delta_{kl},
\end{equation}
where ``P. B.'' stands for the Poisson bracket and \(\frac{\cev \partial}{\partial \xi_j}\) and \(\frac{\vec \partial}{\partial \xi_j}\) are right and left derivatives respectively, as defined in~\cite{Kocia17_2}.

To quantize our algebra, we replace the Poisson brackets for the canonical variables by the anti-commutator multiplied by \(-i/\hbar\)~\cite{Berezin77}:
\begin{equation}
  \{\xi_k, \xi_l\}_{\text{P.B.}} \rightarrow \{\hat \xi_k, \hat \xi_l\} = \hbar \delta_{kl}.
\end{equation}
Renormalizing, we get the Clifford algebra with the three generators:
\begin{equation}
  \label{eq:GrassmannPauli}
  \hat \xi_k = \sqrt{\frac{\hbar}{2}} \hat \sigma_k.
\end{equation}
These \(\hat \sigma_k\) are the Pauli operators.

It can be shown that the operator
\begin{equation}
  \label{eq:Grassmanntranslationop}
  \hat T(\bs \rho) = \exp \left( \frac{2 i}{\hbar} \sum_k \hat \xi_k \rho_k\right)
\end{equation}
corresponds to a translation operator and the dual to the translation operator \(\hat T\) is:
\begin{equation}
  \hat R(\bs \xi) = \int \exp\left(- \frac{2 i}{\hbar} \sum_k \xi_k \rho'_k \right) \hat T(\bs \rho') \text d^3 \rho',
\end{equation}
which corresponds to a reflection (actually an inversion) operator.

As in the odd \(d\) case in Section~\ref{sec:odddqudits}, these reflections are parametrized by phase space \(\bs \xi \equiv (\xi_p, \xi_q, \xi_r)\), which here is made up of elements of the Grassmann algebra \(\mathcal G_3\) instead of \(\mathbb Z/ d \mathbb Z\).

These \(\hat R\) serve as a complete operator basis for any Hilbert space operator \(\hat g\) under Grassmann integration~\cite{Kocia17_2}:
\begin{equation}
  \label{eq:operatorGrassmannRdecomp}
  \hat g = \int \hat R(\bs \xi)  g(\bs \xi) \text d^3 \xi.
\end{equation}

Therefore, any operator \(\hat g\) can be expressed as a linear combination of \(\hat R\) operators parametrized by \(\bs \xi\). We identify \(g(\bs \xi)\) as the Weyl symbol of \(\hat g\), or the Wigner function if \(\hat g\) is a density matrix (\(\hat \rho\)). The Weyl symbol \(g(\bs \xi)\) corresponding to the operator \(\hat g\) can be equivalently represented by a sum of even or odd powers of Grassmann elements. Here we choose \(g(\bs \xi)\) to contain only even terms. Eq.~\ref{eq:operatorGrassmannRdecomp} means that the Weyl symbol (or Wigner function) \(g(\bs \xi)\) can be associated with the coefficients making up \(\hat g\)'s decomposition in the \(\hat R\) operator basis.

It remains to define measurement in this three-generator Grassmann WWM formalism. As in the odd \(d\) qudit case, let us begin by considering taking measurements of \(\bs \xi_p\), \(\bs \xi_q\) and now also \(\bs \xi_r\) separately. In the odd \(d\) case we could accomplish this by tracing away the other degrees of freedom. This is not possible for \(d=2\) because of the Grassmann elements. Unlike its two-generator analog, a three-generator Weyl symbol cannot generally produce scalar values after partial traces; it is a map to \(\mathcal G_3\) after all, not \(\mathbb R\). To produce a real value, a three-generator Weyl symbol must be traced over all of its three degrees of freedom.

Fortunately, we can appeal to what we discovered after taking partial traces in odd \(d\): partial traces in \(p\) or \(q\) are the same as a full trace over the Wigner functions of the associated eigenstates of \(\hat p\) and \(\hat q\). In other words, marginals can be obtained as a special case of expectation values.

For instance, for \(d=2\) we can find that the one-qubit state
\begin{equation}
  \hat \rho = \ketbra{\Psi}{\Psi} = \frac{1}{2} \left(1 + \left( \alpha i \hat \xi_r \hat \xi_q + \beta i \hat \xi_p \hat \xi_q + \gamma i \hat \xi_p \hat \xi_r \right) \right)
\end{equation}
has the corresponding Weyl symbol
\begin{equation}
  \label{eq:qubitWeylsymbol}
  \rho'(\xi) = \frac{1}{2} \left(1 + \left( \alpha i \xi_r \xi_q + \beta i \xi_p \xi_q + \gamma i \xi_p \xi_r \right) \right).
\end{equation}
Using the Grassmann integral equations, it is easy to see that taking the trace with the odd Weyl symbols of the \(\hat q\) eigenstates, \(q_{\pm}(\xi) = \frac{1}{2}(i\xi_p \xi_r \xi_q \pm \xi_q)\), produces
\begin{eqnarray}
  Q[\Psi](q) &=& 2 i \int \rho'(\xi) q_{\pm}(\xi) \text d \xi_r \text d \xi_q \text d \xi_p \nonumber\\
  &=&\begin{cases}\left|\Psi(0)\right|^2 = \frac{1}{2}(1+\alpha) & \text{for \(-\),}\\ \left|\Psi(1)\right|^2 = \frac{1}{2}(1-\alpha) & \text{for \(+\)}.\end{cases}
\end{eqnarray}
These results together provide the marginal \(|\Psi(q)|^2\).

In general, the expectation values of the projectors onto the eigenstates of an observable simply give the marginal distribution of that observable.

As before for odd \(d\), we can interpret measurements in terms of indicator functions on discrete phase space. In the example above, the indicator function is \(q_\pm(\xi)\)---the Weyl symbol of the orthogonal set of projectors of \(\hat Q\). This allows us to measure the ontic state \((\xi_p, \xi_q, \xi_r)\) as one would for a classical particle.

Again, as before for the odd \(d\) WWM formalism, we are interested in observables that are products of observables, and their \(\hbar\) expansion. We can begin with~\cite{Berezin77}:
\begin{eqnarray}
  \label{eq:qubitprodofweylsymb1}
  &&W_{A B}(\bs \xi)\\
&=& \left(\frac{\hbar}{2}\right)^3 \int \text d^3 \bs \xi_1 \text d^3 \bs \xi_2 e^{\frac{2}{\hbar} (\xi_1 \xi_2 + \xi_2 \xi + \xi \xi_1) } W_A(\bs \xi_1) W_B(\bs \xi_2), \nonumber
\end{eqnarray}
and let \(\Sigma = \xi_1 + \xi_2\) and \(\sigma = \xi_1 - \xi_2\) so that we can reexpress this equation in terms of a quadratic argument:
\begin{eqnarray}
  &&W_{A B}(\bs \xi) \nonumber\\
  &=& \left(\frac{\hbar}{2}\right)^3 \int \text d^3 \bs \Sigma \text d^3 \bs \sigma e^{\frac{1}{\hbar} \left[(\Sigma + \sigma)(\Sigma - \sigma) + \bs \xi \sigma \right]} \\
  && \qquad \qquad \times W_A\left(\frac{1}{2}(\bs \Sigma + \bs \sigma)\right) W_B\left(\frac{1}{2}(\bs \Sigma - \bs \sigma)\right) \nonumber\\
  &=& \left(\frac{\hbar}{2}\right)^3 \int \text d^3 \bs \Sigma \text d^3 \bs \sigma e^{\frac{1}{\hbar} \left(\bs \Sigma \bs \Sigma - \bs \sigma \bs \sigma + \bs \xi \bs \sigma \right)} \nonumber \\
  && \qquad \qquad \times W_A\left(\frac{1}{2}(\bs \Sigma + \bs \sigma)\right) W_B\left(\frac{1}{2}(\bs \Sigma - \bs \sigma)\right) . \nonumber
\end{eqnarray}
Evaluating this integral by stationary phase, we find the stationary points of the phase \(\phi \equiv \frac{1}{\hbar} \left[(\Sigma + \sigma)(\Sigma - \sigma) + \bs \xi \sigma \right]\) to be
\begin{equation}
  \frac{\vec \partial}{\partial \Sigma}\phi = 2 \sigma = 0 \Leftrightarrow \bs \xi_1 = \bs \xi_2,
\end{equation}
and
\begin{equation}
  \frac{\vec \partial }{\partial \sigma} \phi= 2 \Sigma + 2 \bs \xi = 0 \Leftrightarrow \bs \xi_1 = \bs \xi_2 = \bs \xi.
\end{equation}

The resultant two Grassmann Gaussian integrals~\cite{Berezin77,Kocia17_2} produce the prefactor
\begin{equation}
  \left(\sqrt{\det \left(2 \frac{1}{\hbar} \hat I_3 \right)} \right)^2,
\end{equation}
and so the order \(\hbar^0\) term consists of the prefactor multiplied by the full equation evaluated at the stationary phase point \(\bs \xi_1 = \bs \xi_2 = \bs \xi\):
\begin{equation}
  W_{A B}(\bs \xi) = W_A(\bs \xi) W_B(\bs \xi) + \mathcal O(\hbar).
\end{equation}

We leave a more detailed development of the WWM formalism to the literature~\cite{Berezin77,Almeida98,Rivas99,Rivas00,Kocia16,Kocia17_2} so as not to deviate from our focus on contextuality. To summarize, the main results of interest to us in the finite-dimensional WWM formalism are the Weyl symbol of observables which are products of observables, \(W_{A B}(\bs \zeta)\), and serve to define the contexts of a measurement. These can be expanded with respect to \(\hbar\) to produce the leading term:
\begin{equation}
  \label{eq:quditgroen}
  W_{A B}(\bs x) = W_A(\bs x) W_B(\bs x) + \mathcal O(\hbar),
\end{equation}
for odd \(d\), and
\begin{equation}
  \label{eq:qubitgroen}
  W_{A B}(\bs \xi) = W_A(\bs \xi) W_B(\bs \xi) + \mathcal O(\hbar),
\end{equation}
for \(d=2\). \(W_A(\bs \zeta)\) and \(W_B(\bs \zeta)\) correspond to the Weyl symbols of the observables \(\hat A\) and \(\hat B\) separately and \(W_{AB}(\bs \zeta)\) is the Weyl symbol of the product \(\hat A \hat B\) (where \(\hat A \hat B\) is an obervable iff. \([\hat A, \hat B]=0\)). In the classical limit (\(\hbar \rightarrow 0\)) the Weyl symbols of the products of observables become the product of the Weyl symbols of each operator separately. These identities will prove to be very illuminating for studying measurement contextuality in the next section.

\section{Measurement Contextuality}
\label{sec:meascontext}
Measurement contextuality was defined in the Introduction as the phenomenon described by a rank \(n \ge 2\) observable that can be measured jointly with either one of two other non-commuting observables---two contexts of measurement---and whose outcomes depend on this choice. Measurement contextuality can also be reexpressed in terms of a hidden variable theory (HVT): if the outcome of a measurement on a state computed by an HVT, which can be measured jointly with either one of two other non-commuting observables, is independent of which of the two other observables it is measured jointly with, then the measurement is non-contextual in that HVT.

The WWM formalism described for odd-dimensional qudits in Section~\ref{sec:odddqudits} and for qubits in Section~\ref{sec:qubits} is an example of an HVT. The ``hidden variables'' in this HVT are the Weyl phase space points, \(\bs x_p\) and \(\bs x_q\) in the odd-dimensional WWM formalism and the stabilizers in the qubit WWM formalism. Every state can be described by these hidden variables, namely by their support on these variables. Stabilizer states can be shown to have non-negative support on these hidden variables and so can be described by \emph{bona fide} classical probability distributions that propagate amongst themselves under Clifford gates~\cite{Kocia16,Kocia17}. This process can be described as one requiring only the \(\mathcal O(\hbar^0)\) terms in its path integral formalism in WWM, and so is preparation non-contextual. Here we turn our attention to measurement contextuality in the WWM formalism.

\begin{theorem}[Measurement Contextuality]
  \label{th:meascontext}
  A pure state \(\hat \rho\equiv \ketbra{\Psi}{\Psi}\) exhibits measurement contextuality under measurement by some observable \(\hat \Sigma\) under contexts \(\hat \Sigma \hat \Sigma_k\) if the Wigner function of the operators, \(W_{\Sigma \Sigma_k}\), must be treated at an order higher than \(\hbar^0\) to compute the expectation values:
    \begin{equation}
      \expval{\Psi}{\hat \Sigma \hat \Sigma_k }{\Psi} = \begin{cases}\int^\infty_{-\infty} W_{\Sigma \Sigma_k}(\bs \xi) \tilde W_{\rho}(\bs \xi) \text d \bs \xi. & \text{for \(d=2\)},\\
        \sum_{\bs x} W_{\Sigma \Sigma_k}(\bs x) W_{\rho}(\bs x)& \text{for odd \(d\)}.\end{cases}
    \end{equation}
\end{theorem}

\begin{proof}
  Let us consider the non-commuting observables, \(\hat \Sigma \hat \Sigma_k\), corresponding to the different contexts of measuring \(\hat \Sigma\). If the measurement is non-contextual then this means that the hidden variables predicting the expectation values of \(\hat \Sigma\) are the same whether \(\hat \Sigma\) is measured in the context with \(\hat \Sigma_k\) or with another \(\hat \Sigma_{k'}\). In other words, the indicator function for \(\hat \Sigma\) must be independent of the indicator function for \(\hat \Sigma_k\) and \(\hat \Sigma_{k'}\):
  \begin{widetext}
    \begin{equation}
      \expval{\Psi}{\hat \Sigma \hat \Sigma_k}{\Psi} =
      \label{eq:th1}
      \begin{cases}
      \int^\infty_{-\infty} W_{\Sigma \Sigma_k}(\bs \xi) \tilde W_{ \rho}(\bs \xi) \text d \bs \xi \underset{\text{non-context.}}{=} \int^\infty_{-\infty} W_{ \Sigma}(\bs \xi) W_{ \Sigma_k}(\bs \xi) \tilde W_{ \rho}(\bs \xi)  \text d \bs \xi & \text{for \(d=2\)},\\
      \sum_{\bs x} W_{ \Sigma  \Sigma_k}(\bs x) W_{ \rho}(\bs x) \underset{\text{non-context.}}{=} \sum_{\bs x} W_{ \Sigma}(\bs x) W_{ \Sigma_k}(\bs x) W_{ \rho}(\bs x) & \text{for odd \(d\)},
      \end{cases}
    \end{equation}
  \end{widetext}
for all given contexts \(k\).  By ``\(\underset{\text{non-context.}}{=}\)'' we denote equality given non-contextual measurements; the left-hand side equals the right-hand side given that the corresponding measurements are non-contextual. (We note that this neither implies that \(\expval{\Psi}{\hat \Sigma \hat \Sigma_k}{\Psi} = \expval{\Psi}{\hat \Sigma \hat \Sigma_{k'}}{\Psi}\) nor that \(\expval{\Psi}{\hat \Sigma \hat \Sigma_{k/k'}}{\Psi} = \expval{\Psi}{\hat \Sigma}{\Psi} \expval{\Psi}{\hat \Sigma_{k/k'}}{\Psi}\)).

  Eq.~\ref{eq:th1} is exactly equal to the \(\hbar^0\) limit of products of observables (Eq.~\ref{eq:quditgroen} for odd \(d\) qudits and Eq.~\ref{eq:qubitgroen} for qubits).

Therefore, only if the \(\hbar\) expansion of \(W_{\Sigma \Sigma_k}\) must be treated at an order higher than order \(\hbar^0\) is the associated measurement \(\expval{\Psi}{\hat \Sigma \hat \Sigma_{k}}{\Psi}\) contextual.\qed
\end{proof}

The approach used in the proof, relies on products of observables (Hermitian operators). This permits the relatively simple treatment presented here instead of using a more involved method likely involving a semiclassical treatment of the Lindblad equation to capture the non-unitary measurement process~\cite{Weinberg16}.

\section{Examples}

The two examples of state-independent contextuality for qubits and state-dependent contextuality for qutrits examined in the Introduction have become well known because of their particular simplicity. We reexamine them here with regards to their expansion in terms of \(\hbar\) and Theorem~\ref{th:meascontext}.

\subsection{Peres-Mermin Square}

It is possible to see the contextuality of the qubit Pauli operators in the Peres-Mermin square from the perspective of Theorem~\ref{th:meascontext}. Multiplying together any of the operators in a row or column corresponding to a context requires multiplying two Pauli operators \(\hat \sigma_1\) and \(\hat \sigma_2\) in each qubit tensor factor. For instance, in the first column we can assign \(\hat \sigma_1 = \hat X\) and \(\hat \sigma_2 = \hat X\) for the first qubit from the entries in the first and third rows, and we can assign \(\hat \sigma_1 = \hat Z\) and \(\hat \sigma_2 = \hat Z\) for the second qubit from the entries in the second and third rows (the identity matrices act trivially). Examining each qubit subspace separately, we find that the corresponding Weyl symbol of these observables is zero at order \(\hbar^0\) since the square of Grassmann elements is equal to zero:
\begin{widetext}
\begin{eqnarray}
  W_{\hat \sigma_a \hat \sigma_b}(\bs \xi) &=& \left( \alpha_1 i \xi_r \xi_q + \beta_1 i \xi_p \xi_q + \gamma_1 i \xi_p \xi_r \right) \left( \alpha_2 i \xi_r \xi_q + \beta_2 i \xi_p \xi_q + \gamma_2 i \xi_p \xi_r \right) + \mathcal O (\hbar)\\
                    &=& 0 + \mathcal O (\hbar). \nonumber
\end{eqnarray}
\end{widetext}

In this way, it is clear from Theorem~\ref{th:meascontext} that Pauli qubit operators are contextual for all states, as they all require the order \(\hbar^1\) term of the measurement operator.

Note that this is generally not true for odd \(d\). The product of two odd \(d\) Weyl symbols is only equal to zero if both are equal to zero; the result for qubits is a unique property dependent on the Grassmann algebra that underpins them. This means that odd \(d\) qudit observables will always have a finite order \(\hbar^0\) term as we shall see in the following examination of the KCSB construction.

\subsection{KCSB}

The fact that \(\hat \Sigma_{\Gamma^2}\) is a sum of products of observables allows us to directly make use of Theorem~\ref{th:meascontext} by examining its Weyl expectation value:
\begin{eqnarray}
  W_{\Sigma_{\Gamma^2}}(\bs x) &=& \sum_{\hat \Pi_i \hat \Pi_j \in \Gamma^2} W_{\Pi_i}(\bs x) W_{\Pi_j}(\bs x) + \mathcal O (\hbar) \nonumber\\
                               &\equiv& W_{\Sigma_{\Gamma^2}}^{\hbar^0}(\bs x) + W_{\Sigma_{\Gamma^2}}^{\hbar}(\bs x), \nonumber
\end{eqnarray}
where all order \(\hbar^0\) terms are defined to be in \(W_{\Sigma_{\Gamma^2}}^{\hbar^0}(\bs x)\) and all higher order terms in \(W_{\Sigma_{\Gamma^2}}^{\hbar}(\bs x)\).

These expectation values of \(\hat \Sigma_{\Gamma^2}\) and its order \(\hbar^0\) and \(\hbar\) parts are given in Table~\ref{tab:KCSBexpvals}, evaluated with all the qutrit stabilizer states. From this table we find that the expectation values of the two states \(\ket{\phi_1}\) and \(\ket{\phi_2}\) with \(\hat \Sigma_{\Gamma^2}\) are largely captured at order \(\hbar^0\) while the third eigenstate of \(\hat \Sigma_\Gamma\), \(\ket{\phi_3}\), has an order \(\hbar^1\) term that is dominant. Indeed, it is even greater than the classical bound and is necessary for \(\ket{\phi_3}\)'s expectation value with \(\hat \Sigma_{\Gamma^2}\) to surpass its classical bound. As a result, by Theorem~\ref{th:meascontext}, we conclude that the first two states exhibit measurement non-contextuality while the third exhibits measurement contextuality under the KCSB construction. This agrees with the conclusion reached from traditional outcome assigment argument presented in the Introduction.

\begin{table*}[ht]
  \begin{tabular}{|c|c|c|c|}
    \hline
    \(\ket \phi\) & \(\sum_{\bs x} W_\phi(\bs x) W_{\Sigma_{\Gamma^2}}(\bs x)\) & \(\sum_{\bs x} W_\phi(\bs x) W_{\Sigma_{\Gamma^2}}^{\hbar^0}(\bs x)\) & \(\sum_{\bs x} W_\phi(\bs x) W_{\Sigma_{\Gamma^2}}^{\hbar}(\bs x)\)\\
    \hline
    \(\ket{\phi_1}\) & \(5-2 \sqrt{5}\) & \(\frac{1}{12} \left(25-9 \sqrt{5}\right)\) & \(\frac{5}{12} \left(7-3 \sqrt{5}\right) \approx 0.12\)\\
    \(\ket{\phi_2}\) & \(5-2 \sqrt{5}\) & \(\frac{1}{12} \left(25-9 \sqrt{5}\right)\) & \(\frac{5}{12} \left(7-3 \sqrt{5}\right) \approx 0.12\)\\
    \(\ket{\phi_3}\) & \(5-\sqrt{5}\) & \(\frac{1}{6} \left(5-\sqrt{5}\right)\) & \(\frac{5}{6} \left(5-\sqrt{5}\right) \approx \bs{2.30}\)\\
    \(\frac{1}{\sqrt{3}}\left(\ket{\phi_1} + \ket{\phi_2} + \ket{\phi_3} \right)\) & \(5-\frac{5 \sqrt{5}}{3}\) & \(\frac{43}{18}-\frac{5 \sqrt{5}}{6}\) & \(\frac{47}{18}-\frac{5 \sqrt{5}}{6} \approx 0.75\)\\
    \(\frac{1}{\sqrt{3}}\left(\ket{\phi_1} + e^{\frac{2 \pi}{3} i}\ket{\phi_2} + e^{\frac{4 \pi}{3} i}\ket{\phi_3} \right)\) & \(5-\frac{5 \sqrt{5}}{3}\) & \(\frac{1}{36} \left(47-15 \sqrt{5}\right)\) & \(\frac{133}{36}-\frac{5 \sqrt{5}}{4} \approx 0.90\)\\
    \(\frac{1}{\sqrt{3}}\left(\ket{\phi_1} + e^{\frac{4 \pi}{3} i}\ket{\phi_2} + e^{\frac{2 \pi}{3} i}\ket{\phi_3} \right)\) & \(5-\frac{5 \sqrt{5}}{3}\) & \(\frac{1}{36} \left(47-15 \sqrt{5}\right)\) & \(\frac{133}{36}-\frac{5 \sqrt{5}}{4} \approx 0.90\)\\
    \(\frac{1}{\sqrt{3}}\left(\ket{\phi_1} + e^{\frac{4 \pi}{3} i}\ket{\phi_2} + e^{\frac{4 \pi}{3} i}\ket{\phi_3} \right)\) & \(5-\frac{5 \sqrt{5}}{3}\) & \(\frac{1}{36} \left(65-21 \sqrt{5}\right)\) & \(\frac{1}{36} \left(115-39 \sqrt{5}\right) \approx 0.77\)\\
    \(\frac{1}{\sqrt{3}}\left(e^{\frac{4 \pi}{3} i}\ket{\phi_1} + e^{\frac{4 \pi}{3} i}\ket{\phi_2} + \ket{\phi_3} \right)\) & \(5-\frac{5 \sqrt{5}}{3}\) & \(\frac{1}{18} \left(25-9 \sqrt{5}\right)\) & \(\frac{65}{18}-\frac{7 \sqrt{5}}{6} \approx 1.00\)\\
    \(\frac{1}{\sqrt{3}}\left(e^{\frac{4 \pi}{3} i}\ket{\phi_1} + \ket{\phi_2} + e^{\frac{4 \pi}{3} i}\ket{\phi_3} \right)\) & \(5-\frac{5 \sqrt{5}}{3}\) & \(\frac{1}{36} \left(65-21 \sqrt{5}\right)\) & \(\frac{1}{36} \left(115-39 \sqrt{5}\right) \approx 0.77\)\\
    \(\frac{1}{\sqrt{3}}\left(e^{\frac{2 \pi}{3} i}\ket{\phi_1} + e^{\frac{2 \pi}{3} i}\ket{\phi_2} + \ket{\phi_3} \right)\) & \(5-\frac{5 \sqrt{5}}{3}\) & \(\frac{1}{18} \left(25-9 \sqrt{5}\right)\) & \(\frac{65}{18}-\frac{7 \sqrt{5}}{6} \approx 1.00\)\\
    \(\frac{1}{\sqrt{3}}\left(\ket{\phi_1} + e^{\frac{2 \pi}{3} i}\ket{\phi_2} + e^{\frac{2 \pi}{3} i}\ket{\phi_3} \right)\) & \(5-\frac{5 \sqrt{5}}{3}\) & \(\frac{1}{36} \left(65-21 \sqrt{5}\right)\) & \(\frac{1}{36} \left(115-39 \sqrt{5}\right) \approx 0.77\)\\
    \(\frac{1}{\sqrt{3}}\left(e^{\frac{2 \pi}{3} i}\ket{\phi_1} + \ket{\phi_2} + e^{\frac{2\pi}{3} i}\ket{\phi_3} \right)\) & \(5-\frac{5 \sqrt{5}}{3}\) & \(\frac{1}{36} \left(65-21 \sqrt{5}\right)\) & \(\frac{1}{36} \left(115-39 \sqrt{5}\right) \approx 0.77\)\\
    \hline
  \end{tabular}
  \caption{\(\hat \Sigma_{\Gamma^2}\) expectation values of the stabilizer states indicated by the first row. The full (exact) expectation value is given in the second column, the expectation value up to order \(\hbar^0\) is given in the third column, and the difference (\(\sum_{\bs x} W_\phi(\bs x) W_{\Sigma_{\Gamma^2}}^{\hbar}(\bs x)\)) of the two, corresponding to the contribution of order \(\hbar\) to the expectation value is given in the fourth column. The order \(\hbar\) contribution in the third row, corresponding to the state \(\ket{\phi_2}\), is shown in bold to highlight that the contribution of order \(\hbar\) to its expectation value is greater than the classical bound on the expectation value.}
  \label{tab:KCSBexpvals}
\end{table*}

\section{Discussion}

In the example involving the Peres-Mermin square, we found that the state-dependent measurement contextuality could be attributed to the fact that qubit Pauli observables have no order \(\hbar^0\) term. The lack of the order \(\hbar^0\) term is unique to \(d=2\). By contradistinction, odd-dimensional qubit operators often have a finite order \(\hbar^0\) term. Indeed, this is because a zero order \(\hbar^0\) term is impossible in the WWM formalism with the algebra used for odd \(d\) qudits (for non-zero operators). It is uniquely due to the Grassmann algebra required for qubit WWM that permits the leading order \(\hbar^0\) term to be zero. Indeed, the likely reason for why the relationship between orders of \(\hbar\) and measurement contextuality wasn't noticed, as far as we can tell from the literature, is that the qubit WWM Grassmann formalism was only recently fully formulated~\cite{Berezin77,Kocia17_2}.

The significance of a zero order \(\hbar^0\) term for qubit Pauli observables under WWM is clear: the (pseudo-) classical limit of qubit Pauli observables is the null operator. In other words, there are no classical analogues to qubit Pauli measurements. This should not be too surprising---it has long been noticed that there is no classical analogue to spin-\(\frac{1}{2}\)~\cite{Berezin77}.

This dichotomy in the possibilities of the magnitude of the order \(\hbar^0\) term between qubits and qudits is the reason that qubits exhibit state-\emph{independent} contextuality while odd \(d\) qudits exhibit state-\emph{dependent} contextuality. Colloquially, it can be said that the terms that are higher order than \(\hbar^0\) are responsible for getting a state's expectation value with an observable to be higher than its classical bound. In the qubit case, since the lower order \(\hbar^0\) term is always zero, these higher order terms always dominate and allow qubit Pauli observables to always exhibit contextuality, regardless of the state being measured. In the odd \(d\) qudit case, constructions can be made that favor particular states with large contributions from their higher order terms such that they exhibit contextuality within that construction, whereas other states don't. However, as seen in the KCSB example, these ``non-contextual'' states still require their higher than order \(\hbar^0\) terms to attain their particular ``classical'' expectation values. They are still quantum animals in this respect, but since they don't distinguish themselves with respect to a constructed classical bound, they are said not to exhibit contextuality with that particular measurement in that particular construction. Contextuality is present but insufficient to violate a particular inequality demonstrating non-classicality. This is analogous to the existence of entangled states that nevertheless cannot violate a Bell inequality.

Moreover, the fact that even though every observable has higher than order \(\hbar^0\) terms that are non-zero but all states are not measurement contextual under some constructions reconciles how it is possible for a quasi-probability representation---such as the WWM---to exhibit negativity for either states or measurements or make use of a deformed probability calculus~\cite{Ferrie08}, while still exhibiting measurement non-contextuality in a particular construction. Though the higher order \(\hbar\) quantum corrections introduce a deformed algebra into WWM's measurement indicator functions for the KCSB construction, they are not a large enough deformation for some states to exhibit expectation values that violate their classical bounds.

Fundamentally, these results show that the measurement contextuality as a resource is the same as orders of \(\hbar\) as a resource. Contextuality is a resource that is necessary for universal quantum computation just like orders of \(\hbar\) higher than \(\hbar^0\) are necessary for quantum phenomena that are richer than their classical counterparts. As mentioned, a similar case for preparation contextuality has also been made~\cite{Howard14} with regard to Clifford operations on stabilizer states which have been shown to only require order \(\hbar^0\) terms in their WWM path integral treatment, while extensions that allow for universal quantum computing require order \(\hbar^1\) terms and higher~\cite{Kocia16}.

To drive this point home, let us examine one more example involving the last remaining manifestation of measurement contextuality: state-independent contextuality for a qutrit.

\subsection{Thirteen Ray Qutrit State-Independent Construction}
\label{sec:13rays}
\begin{figure}[ht]
  \includegraphics[scale=1.0]{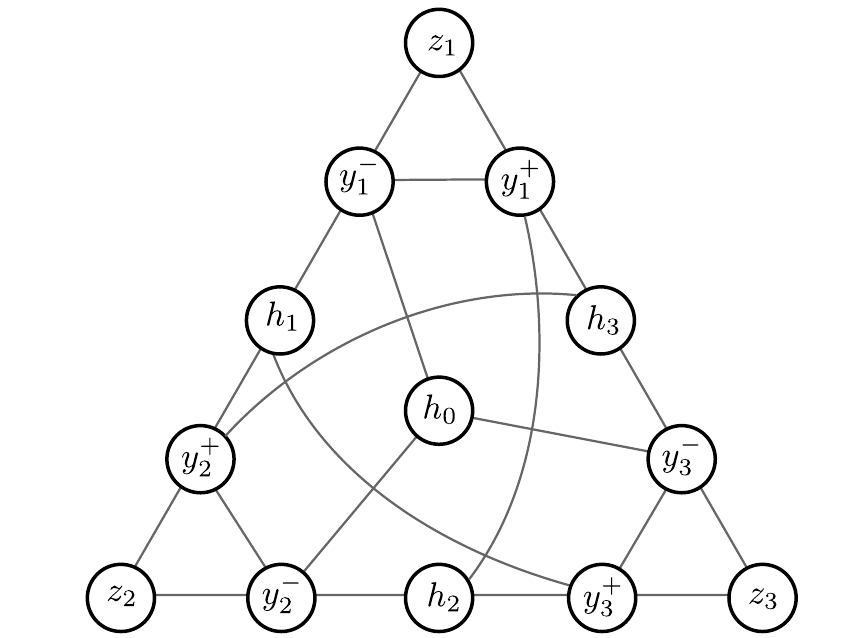}
  \caption{Orthogonality graph of the \(13\) ray qutrit construction.}
  \label{fig:13raysorthograph}
\end{figure}

Consider the following \(13\) vectors or ``rays''~\cite{Yu12}:
\begin{equation}
  \label{eq:13rays}
  \begin{array}{lll}
    y_1^- = (0,1,-1) & h_1 = (-1,1,1) & z_1 = (1,0,0)\\
    y_2^- = (1,0,-1) & h_2 = (1,-1,1) & z_2 = (0,1,0)\\
    y_3^- = (1,-1,0) & h_3 = (1,1,-1) & z_3 = (0,0,1)\\
    y_1^+ = (0,1,1) & h_0 = (1,1,1) & \\
    y_2^+ = (1,0,1) & & \\
    y_3^+ = (1,1,0)& & \\
    \end{array}
  \end{equation}
  The orthogonality of these rays is indicated in Fig.~\ref{fig:13raysorthograph}, where vertices that share an edge (i.e. are joined) indicate that their corresponding rays are orthogonal to each other.

  We let \(V\) denote the set of these rays, such that \(V = \{y_k^\sigma, h_\alpha, z_k | k = 1,2,3; \, \sigma = \pm; \, \alpha = 0,1,2,3\}\), and \(\Gamma\) be the \(13 \times 13\) symmetric adjacency matrix with vanishing diagonals, such that \(\Gamma_{\mu\nu} = 1\) if the two rays \(\mu, \, \nu \in V\) are neighbors and \(\Gamma_{\mu \nu} = 0\) otherwise.

  It follows that for an arbitrary set of \(13\) dichotomic observables, \(\{A_\nu | \nu \in V\}\), each of which takes values (i.e. has outcomes) \(a_\nu^\lambda = \pm 1\) for an ontic state \(\lambda\),
  \begin{equation}
    \label{eq:13rayclinequality}
  \sum_{\nu \in V} \left<A_\nu\right>_{\text{cl}} - \frac{1}{4}\sum_{\mu, \nu \in V} \Gamma_{\mu \nu} \left< A_\mu A_\nu\right>_{\text{cl}} \le 8,
\end{equation}
where \(\left< A_\nu\right>_{\text{cl}} \equiv \int \text d \lambda \rho_\lambda a_\nu^\lambda\) and \(\left<A_\mu A_\nu\right>_{\text{cl}} \equiv \int \text d \lambda \rho_\lambda a_\mu^\lambda a_\nu^\lambda\) for a classical density \(\rho_\lambda\).

Converting this scenario to quantum mechanics, we define \(\hat A_\nu = \hat I - 2 \hat r_\nu\) from the projectors onto the \(13\) rays \(\hat r_\nu \in \{\hat y_k^\sigma, \hat h_\alpha, \hat z_k\}\) (which are outer products of the vectors in Eq.~\ref{eq:13rays}). Observe that \(\hat A_\mu\) and \(\hat A_\nu\) commute whenever their associated rays are orthogonal, i.e., \(\Gamma_{\mu \nu} = 1\). It follows that the quantum analogue of Eq.~\ref{eq:13rayclinequality} is:
\begin{equation}
  \sum_{\nu \in V} \hat A_\nu - \frac{1}{4} \sum_{\mu, \nu \in V} \Gamma_{\mu \nu} \hat A_\mu \hat A_\nu = \frac{25}{3} \hat I.
\end{equation}
Observe that \(\frac{25}{3} > 8\) and so the expectation value of the above equation with any state is always greater than its classical version in Eq.~\ref{eq:13rayclinequality}. This inequality therefore demonstrates state-independent contextuality for one qutrit under this construction.

It is also possible to show state-independent contextuality in this system by considering the simpler inequality:
\begin{equation}
  \label{eq:13rayclinequality2}
  \sum_{\alpha=0}^3 \left< \hat h_\alpha \right>_{\text{cl}} \le 1,
\end{equation}
where \(\left< \hat h_\alpha \right>_{\text{cl}} = \int \text d \lambda \rho_\lambda h_\alpha^\lambda\) for \(h^\lambda_\alpha \in \mathbb R\) since it corresponds to a projector outcome. Unlike the previous inequality in Eq.~\ref{eq:13rayclinequality}, this one additionally relies on the algebraic structure of compatible observables to be preserved classically, i.e., the sum and product rules. In particular, this inequality can be shown to be satisfied based on the following two conditions~\cite{Yu12}:
\begin{enumerate}
\item The value \(\{0,1\}\) assigned to a ray is independent of which bases it finds itself in and
\item One and only one ray is assigned a value of \(1\) among all the rays in a complete orthonormal basis.
\end{enumerate}

This leads to the following two conclusions which together imply the inequality in Eq.~\ref{eq:13rayclinequality2}:
\begin{enumerate}
\item if \(\hat h_0\) and \(\hat h_1\) are assigned to value \(1\) then \(\hat y_2^\pm\) and \(\hat y_3^\pm\) must be assigned to \(0\) so that both \(\hat z_2\) and \(\hat z_3\) must be assigned to value \(1\) which is impossible, and
\item if \(\hat h_1\) and \(\hat h_2\) are assigned to \(1\) then \(\hat y_1^\pm\) and \(\hat y_2^\pm\) must be zero so that both \(\hat z_1\) and \(\hat z_2\) must be assigned to value \(1\), also a contradition.
\end{enumerate}

However, in quantum mechanics,
\begin{equation}
  \hat \Sigma_h = \sum_{\alpha=0}^3 \hat h_\alpha = \frac{4}{3} \hat I,
\end{equation}
and so the expectation value with any quantum state is always \(\frac{4}{3} > 1\), which is greater than the classical bound in Eq.~\ref{eq:13rayclinequality2}, thereby demonstrating state-independent contextuality.

To use this simpler inequality with Theorem~\ref{th:meascontext}, we need to transform it into a sum of bilinear products of observables. This is accomplished easily enough by squaring it:
\begin{equation}
  \label{eq:hcontextwitness2}
  \hat \Sigma_{h^2} = \left(\sum_{\alpha=0}^3 \hat h_\alpha\right)^2 = \frac{16}{9} \hat I.
\end{equation}
This remains a state-independent contextuality witness since its classical upper bound is still \(1\) while the quantum expectation value with any quantum state is \(\frac{16}{9} > 1\).

The order \(\hbar^0\) contribution to the Weyl expectation value of Eq.~\ref{eq:hcontextwitness2} with any qutrit stabilizer state \(\phi\) reveals:
\begin{equation}
  \sum_{\bs x} W^{\hbar^0}_{\Sigma_{h^2}}(\bs x) W_{\phi}(\bs x) = \frac{16}{27} \approx 0.592593,
\end{equation}
which is less than or equal to \(1\), as expected. Since all stabilizer states have the same expectation value, and any other state can be expressed as a linear combination of an orthogonal set of stabilizer states, it follows that every quantum state has the same order \(\hbar^0\) contribution.

As a result, the greater than order \(\hbar^0\) contribution to this expectation value for \emph{all} quantum states is \(\frac{16}{9} - \frac{16}{27} = \frac{32}{27} > 1\), which is a contribution that is not only greater than the order \(\hbar^0\) contribution, but is also larger than the largest possible classical value of \(1\) (the classical bound). Moreover, this contribution must be included in order for the state's Weyl expectation value to violate its classical bound. In this way, Th.~\ref{th:meascontext} shows that the \(13\) ray construction is an illustration of state-independent measurement contextuality for a qutrit system.

\section{Conclusion}

Using the WWM formalism to develop an expansion with respect to \(\hbar\) for products of observables, we showed that states exhibit measurement contextuality if and only if their expectation value for the measurement has a non-zero term higher than order \(\hbar^0\), which must be included in order for the observable to violate a classical bound. This allowed us to show that qubits exhibit state-independent measurement contextuality in the Peres-Mermin square because the associated measurements have no order \(\hbar^0\) term. As a result, any state's expectation value with these observables' associated Weyl symbols must include terms order \(\hbar^1\) or higher and so manifestly violate any constructed classical bound. Conversely, odd-dimensional qudit measurement operators often have have non-zero order \(\hbar^0\) terms and so exhibit state-dependent measurement contextuality. Only if an odd-dimensional qudit state's expectation value with an observable requires the order \(\hbar^1\) term for it to surpass its classical bounds, does it exhibits measurement contextuality.

With this last development, the formal equivalence between contextuality and order of \(\hbar\) is complete; in both the case of preparation and measurement contextuality, states exhibit contextuality under an operator (unitary or Hermitian, respectively) if it must be treated at higher than order \(\hbar^0\) order to capture the result (an evolution and a measurement outcome respectively). Contextuality as a resource is equivalent to orders of \(\hbar\) as a resource. The often far-flung studies of contextuality and semiclassics (roughly the study of the importance of higher order expansions of \(\hbar\)) are in fact concerned with an equivalent phenomenon.

\acknowledgments
This work was supported by AFOSR award no. FA9550-12-1-0046. Parts of this manuscript are a contribution of NIST, an agency of the US government, and are not subject to US copyright.

\bibliography{biblio}{}
\bibliographystyle{unsrt}
\pagebreak
\section*{Appendix: Weyl Symbols for Operators in the KCSB Construction}
We list here the Weyl symbols of the projectors \(\hat \Pi_i\) and the stabilizer states \(\ketbra {\phi}{\phi}\) used in Table~\ref{tab:KCSBexpvals}.

\begin{table}[ht]
  \begin{equation*}
    \ket \phi = \ket{\phi_1}:
  \end{equation*}
  \begin{tabular}{|c|c|c|c|}
    \hline
    \(W_{\phi}(x_p, x_q)\) & \(x_q=0\) & \(x_q=1\) & \(x_q=2\)\\
    \hline
    \(x_p=0\) & \(\frac{1}{3}\)& \(0\)& \(0\)\\
    \(x_p=1\) & \(\frac{1}{3}\)& \(0\)& \(0\)\\
    \(x_p=2\) & \(\frac{1}{3}\)& \(0\)& \(0\)\\
    \hline
  \end{tabular}
  \begin{equation*}
    \ket \phi = \ket{\phi_2}:
  \end{equation*}
  \begin{tabular}{|c|c|c|c|}
    \hline
    \(W_{\phi}(x_p, x_q)\) & \(x_q=0\) & \(x_q=1\) & \(x_q=2\)\\
    \hline
    \(x_p=0\) & \(0\) & \(\frac{1}{3}\)& \(0\)\\
    \(x_p=1\) & \(0\) & \(\frac{1}{3}\)& \(0\)\\
    \(x_p=2\) & \(0\) & \(\frac{1}{3}\)& \(0\)\\
    \hline
  \end{tabular}
  \begin{equation*}
    \ket \phi = \ket{\phi_3}:
  \end{equation*}
  \begin{tabular}{|c|c|c|c|}
    \hline
    \(W_{\phi}(x_p, x_q)\) & \(x_q=0\) & \(x_q=1\) & \(x_q=2\)\\
    \hline
    \(x_p=0\) & \(0\)& \(0\) & \(\frac{1}{3}\)\\
    \(x_p=1\) & \(0\)& \(0\) & \(\frac{1}{3}\)\\
    \(x_p=2\) & \(0\)& \(0\) & \(\frac{1}{3}\)\\
    \hline
  \end{tabular}
  \begin{equation*}
    \ket \phi = \frac{1}{\sqrt{3}}\left(\ket{\phi_1} + \ket{\phi_2} + \ket{\phi_3} \right):
  \end{equation*}
  \begin{tabular}{|c|c|c|c|}
    \hline
    \(W_{\phi}(x_p, x_q)\) & \(x_q=0\) & \(x_q=1\) & \(x_q=2\)\\
    \hline
    \(x_p=0\) & \(0\)& \(0\) & \(0\)\\
    \(x_p=1\) & \(0\)& \(0\) & \(0\)\\
    \(x_p=2\) & \(\frac{1}{3}\)& \(\frac{1}{3}\) & \(\frac{1}{3}\)\\
    \hline
  \end{tabular}
  \begin{equation*}
    \ket \phi = \frac{1}{\sqrt{3}}\left(\ket{\phi_1} + e^{\frac{2 \pi}{3} i}\ket{\phi_2} + e^{\frac{4 \pi}{3} i}\ket{\phi_3} \right):
  \end{equation*}
  \begin{tabular}{|c|c|c|c|}
    \hline
    \(W_{\phi}(x_p, x_q)\) & \(x_q=0\) & \(x_q=1\) & \(x_q=2\)\\
    \hline
    \(x_p=0\) & \(0\)& \(0\) & \(0\)\\
    \(x_p=1\) & \(\frac{1}{3}\)& \(\frac{1}{3}\) & \(\frac{1}{3}\)\\
    \(x_p=2\) & \(0\)& \(0\) & \(0\)\\
    \hline
  \end{tabular}
  \begin{equation*}
    \ket \phi = \frac{1}{\sqrt{3}}\left(\ket{\phi_1} + e^{\frac{4 \pi}{3} i}\ket{\phi_2} + e^{\frac{2 \pi}{3} i}\ket{\phi_3} \right):
  \end{equation*}
  \begin{tabular}{|c|c|c|c|}
    \hline
    \(W_{\phi}(x_p, x_q)\) & \(x_q=0\) & \(x_q=1\) & \(x_q=2\)\\
    \hline
    \(x_p=0\) & \(0\)& \(0\) & \(0\)\\
    \(x_p=1\) & \(0\)& \(0\) & \(0\)\\
    \(x_p=2\) & \(\frac{1}{3}\)& \(\frac{1}{3}\) & \(\frac{1}{3}\)\\
    \hline
  \end{tabular}
  \label{tab:weylsymbolsofstabstates1}
  \caption{Weyl Symbols of first six stabilizer states in Table~\ref{tab:KCSBexpvals}.}
\end{table}

\begin{table}[ht]
  \begin{equation*}
    \ket \phi = \frac{1}{\sqrt{3}}\left(\ket{\phi_1} + e^{\frac{4 \pi}{3} i}\ket{\phi_2} + e^{\frac{4 \pi}{3} i}\ket{\phi_3} \right):
  \end{equation*}
  \begin{tabular}{|c|c|c|c|}
    \hline
    \(W_{\phi}(x_p, x_q)\) & \(x_q=0\) & \(x_q=1\) & \(x_q=2\)\\
    \hline
    \(x_p=0\) & \(\frac{1}{3}\)& \(0\)& \(0\)\\
    \(x_p=1\) & \(0\) & \(\frac{1}{3}\)& \(0\)\\
    \(x_p=2\) & \(0\) & \(0\) & \(\frac{1}{3}\)\\
    \hline
  \end{tabular}
  \begin{equation*}
    \ket \phi = \frac{1}{\sqrt{3}}\left(e^{\frac{4 \pi}{3} i}\ket{\phi_1} + e^{\frac{4 \pi}{3} i}\ket{\phi_2} + \ket{\phi_3} \right):
  \end{equation*}
  \begin{tabular}{|c|c|c|c|}
    \hline
    \(W_{\phi}(x_p, x_q)\) & \(x_q=0\) & \(x_q=1\) & \(x_q=2\)\\
    \hline
    \(x_p=0\) & \(0\) & \(0\)& \(\frac{1}{3}\)\\
    \(x_p=1\) & \(\frac{1}{3}\) & \(0\)& \(0\)\\
    \(x_p=2\) & \(0\) & \(\frac{1}{3}\)& \(0\)\\
    \hline
  \end{tabular}
  \begin{equation*}
    \ket \phi = \frac{1}{\sqrt{3}}\left(e^{\frac{4 \pi}{3} i}\ket{\phi_1} + \ket{\phi_2} + e^{\frac{4 \pi}{3} i}\ket{\phi_3} \right):
  \end{equation*}
  \begin{tabular}{|c|c|c|c|}
    \hline
    \(W_{\phi}(x_p, x_q)\) & \(x_q=0\) & \(x_q=1\) & \(x_q=2\)\\
    \hline
    \(x_p=0\) & \(0\)& \(\frac{1}{3}\) & \(0\)\\
    \(x_p=1\) & \(0\)& \(0\) & \(\frac{1}{3}\)\\
    \(x_p=2\) & \(\frac{1}{3}\)& \(0\) & \(0\)\\
    \hline
  \end{tabular}
  \begin{equation*}
    \ket \phi = \frac{1}{\sqrt{3}}\left(e^{\frac{2 \pi}{3} i}\ket{\phi_1} + e^{\frac{2 \pi}{3} i}\ket{\phi_2} + \ket{\phi_3} \right):
  \end{equation*}
  \begin{tabular}{|c|c|c|c|}
    \hline
    \(W_{\phi}(x_p, x_q)\) & \(x_q=0\) & \(x_q=1\) & \(x_q=2\)\\
    \hline
    \(x_p=0\) & \(0\)& \(0\) & \(\frac{1}{3}\)\\
    \(x_p=1\) & \(0\)& \(\frac{1}{3}\) & \(0\)\\
    \(x_p=2\) & \(\frac{1}{3}\)& \(0\) & \(0\)\\
    \hline
  \end{tabular}
  \begin{equation*}
    \ket \phi = \frac{1}{\sqrt{3}}\left(\ket{\phi_1} + e^{\frac{2 \pi}{3} i}\ket{\phi_2} + e^{\frac{2 \pi}{3} i}\ket{\phi_3} \right):
  \end{equation*}
  \begin{tabular}{|c|c|c|c|}
    \hline
    \(W_{\phi}(x_p, x_q)\) & \(x_q=0\) & \(x_q=1\) & \(x_q=2\)\\
    \hline
    \(x_p=0\) & \(\frac{1}{3}\)& \(0\) & \(0\)\\
    \(x_p=1\) & \(0\)& \(0\) & \(\frac{1}{3}\)\\
    \(x_p=2\) & \(0\)& \(\frac{1}{3}\) & \(0\)\\
    \hline
  \end{tabular}
  \begin{equation*}
    \ket \phi = \frac{1}{\sqrt{3}}\left(e^{\frac{2 \pi}{3} i}\ket{\phi_1} + \ket{\phi_2} + e^{\frac{2\pi}{3} i}\ket{\phi_3} \right):
  \end{equation*}
  \begin{tabular}{|c|c|c|c|}
    \hline
    \(W_{\phi}(x_p, x_q)\) & \(x_q=0\) & \(x_q=1\) & \(x_q=2\)\\
    \hline
    \(x_p=0\) & \(0\)& \(\frac{1}{3}\) & \(0\)\\
    \(x_p=1\) & \(\frac{1}{3}\)& \(0\) & \(0\)\\
    \(x_p=2\) & \(0\)& \(0\) & \(\frac{1}{3}\)\\
    \hline
  \end{tabular}
  \label{tab:weylsymbolsofstabstates2}
  \caption{Weyl Symbols of last six stabilizer states in Table~\ref{tab:KCSBexpvals}.}
\end{table}

\begin{table*}[!htp]
  \begin{tabular}{|c|c|c|c|}
    \hline
    \(W_{\Pi_1}(x_p, x_q)\) & \(x_q=0\) & \(x_q=1\) & \(x_q=2\)\\
    \hline
    \(x_p=0\) & \(\frac{1}{60} \left(\sqrt{5}+4 \sqrt{5 \left(3 \sqrt{5}-5\right)}+5\right) \approx 0.32\)& \(\frac{1}{60} \left(-5 \sqrt{5}-4 \sqrt{5 \left(\sqrt{5}+1\right)}+15\right) \approx -0.20 \)& \(\frac{1}{15} \left(\sqrt{5}-5 \sqrt{\frac{2}{\sqrt{5}+5}}\right) \approx -0.03\)\\
    \(x_p=1\) & \(\frac{1}{60} \left(\sqrt{5}-2 \sqrt{5 \left(3 \sqrt{5}-5\right)}+5\right) \approx 0.02\)& \(\frac{1}{60} \left(-5 \sqrt{5}+2 \sqrt{5 \left(\sqrt{5}+1\right)}+15\right) \approx 0.20\)& \(\frac{\sqrt{10-2 \sqrt{5}}+4}{12 \sqrt{5}} \approx 0.24\)\\
    \(x_p=2\) & \(\frac{1}{60} \left(\sqrt{5}-2 \sqrt{5 \left(3 \sqrt{5}-5\right)}+5\right) \approx 0.02\)& \(\frac{1}{60} \left(-5 \sqrt{5}+2 \sqrt{5 \left(\sqrt{5}+1\right)}+15\right) \approx 0.20 \)& \(\frac{\sqrt{10-2 \sqrt{5}}+4}{12 \sqrt{5}} \approx 0.24\)\\
    \hline
  \end{tabular}

  \begin{tabular}{|c|c|c|c|}
    \hline
    \(W_{\Pi_2}(x_p, x_q)\) & \(x_q=0\) & \(x_q=1\) & \(x_q=2\)\\
    \hline
    \(x_p=0\) & \(\frac{1}{15} \left(5-\sqrt{5}\right) \approx 0.18\)& \(\frac{2}{3} \sqrt{\frac{1}{\sqrt{5}}-\frac{1}{5}} \approx 0.33\)& \(\frac{1}{3 \sqrt{5}} \approx 0.15\)\\
    \(x_p=1\) & \(\frac{1}{15} \left(5-\sqrt{5}\right) \approx 0.18\)& \(-\frac{1}{3} \sqrt{\frac{1}{\sqrt{5}}-\frac{1}{5}} \approx -0.17\)& \(\frac{1}{3 \sqrt{5}} \approx 0.15\)\\
    \(x_p=2\) & \(\frac{1}{15} \left(5-\sqrt{5}\right) \approx 0.18\)& \(-\frac{1}{3} \sqrt{\frac{1}{\sqrt{5}}-\frac{1}{5}} \approx -0.17\)& \(\frac{1}{3 \sqrt{5}} \approx 0.15\)\\
    \hline
  \end{tabular}

  \begin{tabular}{|c|c|c|c|}
    \hline
    \(W_{\Pi_3}(x_p, x_q)\) & \(x_q=0\) & \(x_q=1\) & \(x_q=2\)\\
    \hline
    \(x_p=0\) & \(\frac{1}{60} \left(\sqrt{5}-4 \sqrt{5 \left(3 \sqrt{5}-5\right)}+5\right) \approx -0.07\)& \(\frac{1}{60} \left(-5 \sqrt{5}-4 \sqrt{5 \left(\sqrt{5}+1\right)}+15\right) \approx -0.20 \)& \(\frac{\sqrt{10-2 \sqrt{5}}+2}{6 \sqrt{5}} \approx 0.32\)\\
    \(x_p=1\) & \(\frac{1}{60} \left(\sqrt{5}+2 \sqrt{5 \left(3 \sqrt{5}-5\right)}+5\right) \approx 0.22\)& \(\frac{1}{60} \left(-5 \sqrt{5}+2 \sqrt{5 \left(\sqrt{5}+1\right)}+15\right) \approx 0.20\)& \(\frac{1}{3 \sqrt{5}}-\frac{1}{3 \sqrt{2 \left(\sqrt{5}+5\right)}} \approx 0.06\)\\
    \(x_p=2\) & \(\frac{1}{60} \left(\sqrt{5}+2 \sqrt{5 \left(3 \sqrt{5}-5\right)}+5\right) \approx 0.22\)& \(\frac{1}{60} \left(-5 \sqrt{5}+2 \sqrt{5 \left(\sqrt{5}+1\right)}+15\right) \approx 0.20\)& \(\frac{1}{3 \sqrt{5}}-\frac{1}{3 \sqrt{2 \left(\sqrt{5}+5\right)}} \approx 0.06\)\\
    \hline
  \end{tabular}
  
  \begin{tabular}{|c|c|c|c|}
    \hline
    \(W_{\Pi_4}(x_p, x_q)\) & \(x_q=0\) & \(x_q=1\) & \(x_q=2\)\\
    \hline
    \(x_p=0\) & \(\frac{1}{30} \left(2\ 5^{3/4} \sqrt{2}-2 \sqrt{5}+5\right) \approx 0.33\)& \(\frac{1}{30} \left(2 \sqrt{10 \left(\sqrt{5}-2\right)}+5\right) \approx 0.27\)& \(\frac{1}{3} \left(\sqrt{1-\frac{2}{\sqrt{5}}}+\frac{1}{\sqrt{5}}\right) \approx 0.26\)\\
    \(x_p=1\) & \(\frac{1}{30} \left(-5^{3/4} \sqrt{2}-2 \sqrt{5}+5\right) \approx -0.14\)& \(\frac{1}{30} \left(5-\sqrt{10 \left(\sqrt{5}-2\right)}\right) \approx 0.12\)& \(-\frac{\sqrt{5-2 \sqrt{5}}-2}{6 \sqrt{5}} \approx 0.09\)\\
    \(x_p=2\) & \(\frac{1}{30} \left(-5^{3/4} \sqrt{2}-2 \sqrt{5}+5\right) \approx -0.14\)& \(\frac{1}{30} \left(5-\sqrt{10 \left(\sqrt{5}-2\right)}\right) \approx 0.12\)& \(-\frac{\sqrt{5-2 \sqrt{5}}-2}{6 \sqrt{5}} \approx 0.09\)\\
    \hline
  \end{tabular}

  \begin{tabular}{|c|c|c|c|}
    \hline
    \(W_{\Pi_5}(x_p, x_q)\) & \(x_q=0\) & \(x_q=1\) & \(x_q=2\)\\
    \hline
    \(x_p=0\) & \(\frac{1}{30} \left(-2 5^{3/4} \sqrt{2}-2 \sqrt{5}+5\right) \approx -0.30\)& \(\frac{1}{30} \left(2 \sqrt{10 \left(\sqrt{5}-2\right)}+5\right) \approx 0.27\)& \(-\frac{\sqrt{5-2 \sqrt{5}}-1}{3 \sqrt{5}} \approx 0.04\)\\
    \(x_p=1\) & \(\frac{1}{30} \left(5^{3/4} \sqrt{2}-2 \sqrt{5}+5\right) \approx 0.18\)& \(\frac{1}{30} \left(5-\sqrt{10 \left(\sqrt{5}-2\right)}\right) \approx 0.12\)& \(\frac{\sqrt{5-2 \sqrt{5}}+2}{6 \sqrt{5}} \approx 0.20\)\\
    \(x_p=2\) & \(\frac{1}{30} \left(5^{3/4} \sqrt{2}-2 \sqrt{5}+5\right) \approx 0.18\)& \(\frac{1}{30} \left(5-\sqrt{10 \left(\sqrt{5}-2\right)}\right) \approx 0.12\)& \(\frac{\sqrt{5-2 \sqrt{5}}+2}{6 \sqrt{5}} \approx 0.20\)\\
    \hline
  \end{tabular}
  
  \label{tab:weylsymbolsofV}
  \caption{Weyl Symbols of \(\hat \Pi_i\) projectors in Table~\ref{tab:KCSBexpvals}.}
\end{table*}

\end{document}